\documentclass[doublecolumn,journal]{IEEEtran}
\usepackage{times}
\usepackage{cite}
\usepackage{color}
\usepackage{epsf}
\usepackage{epsfig}
\usepackage{graphicx}
\usepackage{graphics}
\usepackage{epstopdf}
\usepackage[small]{caption}
\usepackage{amsmath}
\usepackage{amssymb}
\usepackage{amsthm}
\usepackage{amsxtra}
\usepackage[ruled]{algorithm2e}
\usepackage{algorithmic}
\usepackage{enumerate}
\usepackage{multirow}
\usepackage{enumerate}
\usepackage{amssymb}
\usepackage{lipsum}
\usepackage{setspace}
\usepackage{multirow}
\usepackage{subcaption}
\usepackage{url}

\newtheorem{theorem}{\bf Theorem}

\newtheorem{definition}{\bf Definition}
\newtheorem{remark}{Remark}

\usepackage{rotating} 
\usepackage{graphicx} 

\newcommand{\Rmnum}[1]{\expandafter\@slowromancap\romannumeral #1@}

\makeatother
\begin{document}
\title{Peer-to-Peer Energy Trading with Sustainable User Participation: A Game Theoretic Approach}
\author{Wayes Tushar, Tapan Saha, Chau Yuen, Na Li, Paul Liddell, Richard Bean, H. Vincent Poor, and Kristin Wood}
\author{Wayes~Tushar,~\IEEEmembership{Senior Member,~IEEE,}~Tapan~Kumar~Saha,~\IEEEmembership{Senior Member,~IEEE,}~Chau~Yuen,~\IEEEmembership{Senior Member,~IEEE,}~Paul Liddell, Richard Bean, and~H.~Vincent~Poor,~\IEEEmembership{Fellow,~IEEE}
\thanks{W. Tushar and T. K. Saha are with the School of Information Technology and Electrical Engineering of the University of Queensland, Brisbane, QLD 4072, Australia (e-mail: wayes.tushar.t@ieee.org; saha@itee.uq.edu.au).}
\thanks{C. Yuen is with the Engineering Product Development Pillar of the Singapore University of Technology and Design (SUTD), 8 Somapah Road, Singapore 487372. (e-mail: yuenchau@sutd.edu.sg).}
\thanks{P. Liddell and R. Bean are with the Redback Technologies, Indooroolippy, QLD 4068, Australia. (e-mail: paul@redbacktech.com; richard@redbacktech.com).}
\thanks{H. V. Poor is with the Department of Electrical Engineering at Princeton University, Princeton, NJ 08544, USA. (e-mail: poor@princeton.edu).}
\thanks{This work is supported in part by the Advance Queensland Research Fellowship AQRF11016-17RD2, which is jointly sponsored by the State of Queensland through the Department of Science, Information Technology and Innovation, the University of Queensland, and Redback Technologies, in part by NSFC 61750110529, and in part by the U.S. National Science Foundation under Grant ECCS-1549881.}
}
\IEEEoverridecommandlockouts
\maketitle
\begin{abstract}
This paper explores the feasibility of social cooperation between prosumers within an energy network in establishing their sustainable participation in peer-to-peer (P2P) energy trading. In particular, a canonical coalition game (CCG) is utilized to propose a P2P energy trading scheme, in which a set of participating prosumers form a coalition group to trade their energy, if there is any, with one another. By exploring the concept of the core of the designed CCG framework, the mid-market rate is utilized as a pricing mechanism of the proposed P2P trading to confirm the stability of the coalition as well as to guarantee the benefit to the prosumers for forming the social coalition. The paper further introduces the motivational psychology models that are relevant to the proposed P2P scheme and it is shown that the outcomes of proposed P2P energy trading scheme satisfy the discussed models. Consequently, it is proven that the proposed scheme is consumer-centric that has the potential to corroborate sustainable prosumers participation in P2P energy trading. Finally, some numerical examples are provided to demonstrate the beneficial properties of the proposed scheme.
\end{abstract}
\begin{IEEEkeywords}
Peer-to-peer trading, social cooperation, coalition game, consumer-centric, motivational psychology.
\end{IEEEkeywords}
 \setcounter{page}{1}
\section{Introduction}
\label{sec:introduction}
The global market of rooftop solar photovoltaic (PV) panels, which was US$\$30$ billion in $2016$, is expected to grow by $11\%$ over the next six years. This shift towards solar will further be complemented by an additional increase in residential energy storage systems whose ability to deliver is expected to grow from $95$ MW in $2016$ to $3700$ MW by $2025$~\cite{Peck_Spectrum_2017}. Such additional energy resources at the edge of the grid are expected to be utilized not only to manage the energy demand more efficiently but also to enable a significant mix of clean renewable energy into the grid~\cite{Peck_Spectrum_2017}. However, to make this happen in reality, it is of utmost importance to incorporate the \textit{people who own these generating assets} in their homes, that is, the prosumers~\cite{NLiu_TII_2017},  into the energy market~\cite{Peck_Spectrum_2017}. 

The important role of prosumers in the deregulated energy market is well recognized. For example, the world has already seen the participation of prosumers in the energy market through the feed-in-tariff (FiT) scheme~\cite{Liang-ChengYe-AE:2017}. In FiT, the prosumers with roof-top solar panels can sell their excess solar energy to the grid and can buy energy from the grid in case of any energy deficiency. Unfortunately, the benefit to the prosumers for participating in FiT is not significant \cite{Tushar-TIE:2015}, which has influenced the recent discontinuation of some of the FiT schemes, e.g.,  in Australia \cite{FiTPolicy_2015}. Further, net metering is also used to enable bi-directional trading of energy between prosumers and the grid. However,  increased penetration of non-dispatchable solar energy into the grid can potentially compromise the grid's stability. As a consequence, local governments in many countries have imposed solar export limits on the prosumer.

Given this context, an alternative approach to engaging prosumers in the energy trading market, a proposal on the application of Peer-to-Peer (P2P) trading mechanism in the energy domain is introduced in \cite{Ambrosio_EM_2016}, which can potentially eliminate the limitations of the FiT scheme, and thus contribute towards substantially increasing the percentage of renewable energy penetration into the current electricity grid. Moreover, the P2P trading mechanism not only benefits prosumers economically but also helps the grid to maintain its stability by enabling prosumers to inject to the grid via net metering within the government imposed export threshold. Thus, considering the potential of this P2P trading, a considerable number of pilot projects are being implemented in the USA, Europe, and Australia~\cite{ZhangChenghua_EP_2017}. Nevertheless, the key question that needs to be answered for the successful establishment of a P2P energy trading platform is: how to prepare the P2P trading as a consumer-centric\footnote{The consumer-centric property of a technology, for example, P2P energy trading, properly incentivizes consumers to actively participate in the trading process~\cite{Saad_ProcIEEE_2016}.} scheme that will ensure sustainable participation of prosumers in the P2P energy trading market. Seeking a suitable answer to this question is particularly important due to the following two characteristics of a P2P energy trading market:
\begin{itemize}
\item In P2P trading, the main objective is to encourage the participating prosumers to trade energy with one another with a very low (or, not at all) direct influence from the central controller (for example, the retailer). However, relaxing the influence of the central control body makes the P2P a trustless system~\cite{PowerLedger_2017}. Hence, it would be a challenging task to encourage the prosumers to cooperate to trade energy with one another in such a system.
\item In recent P2P pilot projects such as the Brooklyn microgrid, it is identified that when the prosumers socially interact with one another to exchange their generated solar energy, the trading price per unit of energy may increase substantially~\cite{Mengelkamp_AppliedEnergy_2017}. This may potentially limit the involvement of rational energy users who participate as buyers in the energy market.
\end{itemize} 

Given this context, this paper investigates how sustainable users' participation in P2P energy trading can be established through social cooperation between the prosumers. In particular, we are interested in designing a P2P energy trading technique that encourages the prosumers to participate via forming a coalition group among themselves despite the trustless nature of the system. The trading scheme also needs to be beneficial for the prosumers every time they participate irrespective of their roles, that is, whether they are participating as the buyers or as the sellers. In other words, the trading scheme needs to satisfy the consumer-centric property, where the prosumers would always find it beneficial to participate~\cite{Tushar-TSG:2014}. 

To this end, we have made following contributions in this paper: 1) We propose a P2P trading technique through utilizing the social cooperation between different prosumers within an energy system by designing a Canonical Coalition Game (CCG). By setting the rules of a social coalition with the developed CCG, we show that it is always beneficial to the prosumers of the system, for the considered value function and assumption of this study, to cooperate with one another for trading energy among themselves; 2) By exploring the idea of the core, it is proven that there always exists a revenue distribution that lies within the core, i.e., the core is non-empty. We use a mid-market rate pricing scheme for distributing revenues within the participating prosumers, and show that the revenue that each prosumer receives for forming the coalition lies within the core of the game, and thus ensures \textit{stable} social cooperation; 3) We demonstrate that the proposed scheme is a \textit{consumer-centric technique} that has the capacity to enable significant user participation in energy trading. To do so, we first introduce two behavioral models from motivational psychology which underline the characteristics that have the capability to motivate prosumers to always participate in energy trading. Then, we show that our proposed scheme satisfies both models, and therefore can be considered as a consumer-centric trading scheme; 4) Finally, we validate the properties of the proposed scheme, that is, stability and the consumer-centric property, through real data based numerical simulation.

The remainder of the paper is organized as follows. In Section~\ref{sec:survey}, we provide a survey of existing studies on P2P energy trading. The system model of the proposed work is explained in Section~\ref{sec:system-model} followed by the proposed CCG in Section~\ref{sec:CoalitionGame}. In Section~\ref{sec:properties}, we study the properties of the proposed CCG inspired P2P trading scheme. Some numerical case studies are demonstrated in Section~\ref{sec:CaseStudy}, and finally, we draw some concluding remarks in Section~\ref{sec:Conclusion}.

\section{State-of-the-Art}\label{sec:survey}
Recently, there has been a considerable amount of research effort to understand the potential of P2P trading in the energy sector. The research focus of existing literature can be broadly divided into three general domains: electric vehicles domain, microgrids domain, and distribution networks domain.

\subsection{Electric vehicles domain} In this domain, the main discussion is on the exchange of energy between two sets of electric vehicles, that is, buyer set and seller set, in order to achieve an economically beneficial energy trading platform for all involved electric vehicles in the trading process. Examples of such schemes can be found in \cite{Alvaro-Hermana_ITSM_2016} and\cite{Kang_TII_2017}. In \cite{Alvaro-Hermana_ITSM_2016}, the authors target reducing the impact of the electric vehicle charging process on the power system during business hours. To do so, they exploit an activity-based model to predict the activity of drivers during business hours and then utilize that information to determine an optimal P2P trading price per area and time slot to enable vehicles with excess energy to share their battery energy with vehicles in need of energy during business hours.  Using Consortium Blockchain, the authors in \cite{Kang_TII_2017} propose a secure localized P2P energy trading model for hybrid electric vehicles via a double auction mechanism. It is shown that the proposed trading method can achieve social welfare, improve transaction security, and protect the privacy of the vehicle owners.  
\subsection{Microgrids domain}In the microgrid domain, existing literature on P2P trading focuses on three different trading paradigms: 1) energy trading between different microgrids, 2) energy trading between players within a microgrid, and 3) energy trading between peers and microgrid.
\subsubsection{Trading between different microgrids} Coordinated energy management with networked microgrids has been widely discussed in the literature, for example, \cite{Fathi_TSE_2013, Liu_Smartgridcomm_2015, Moslehi_TPS_2018,Wang_TIE_2016}. In \cite{Fathi_TSE_2013}, a cooperative power dispatching algorithm of interactions among microgrids is proposed for power sharing within the grid. To handle the mismatch between the supply and demand of energy in microgrids, a P2P energy sharing among microgrids is proposed in \cite{Liu_Smartgridcomm_2015} for improving the utilization of distributed energy resources and saving electricity bills for all participating microgrids. In \cite{Moslehi_TPS_2018}, the authors introduce the concept of a nested transactive grid to model the distribution grid as a nested set of virtual microgrids, where each microgrid can act as a market. This facilitates P2P trading while incorporating the security of the grid. Finally, a reinforcement learning based energy trading game among smart microgrids is implemented in \cite{Wang_TIE_2016} that enables each microgrid to individually and randomly choose a strategy to trade the energy in an independent market. 
\subsubsection{Trading between peers within a microgrid} There has also been an increasing interest in modeling P2P trading between prosumers within a microgrid in recent times.  For example, with a target to mitigate the intermittency of renewable generation within microgrids, a concept of distributed generation combined with cooperation by exchanging energy among distributed resources is proposed in \cite{Lakshminarayana_JSAC_2014}. It is shown that in the presence of limited storage devices, the grid can benefit greatly from cooperation, which is reduced in the presence of large storage capacity. In \cite{Liu_TPS_2017}, the authors propose a P2P energy sharing model with price-based demand response, which is shown to be effective in reducing prosumer costs and improving the sharing of photovoltaic energy. Other examples of energy trading between peers within a microgrid can be found in \cite{Lee_JSAC_2014,Tushar_TSG_May_2016,Morstyn_TSG_2018,Werth_TSG_2018,Zhang_TII_2017,Zizzo_TII_2018}.
\subsubsection{Trading between peers and microgrid} An interesting concept of peer-to-microgrid exchange of energy is proposed in \cite{Stevanoni_TSG_2018} with the purpose of long-term planning for connected industrial microgrids. Essentially, the authors propose a new system of daily operation including industrial Load Management and allowing peer-to-microgrid as well as external energy exchanges. The developed tool is tested on a virtual industrial microgrid set up to present the technical and economic outputs. 
\subsection{Distribution networks domain}In this domain, the focus of the studies is to address new challenges for distribution grid operators since they face electricity feed-in at low voltage levels not foreseen when the grid layout was planned~\cite{Thomsen_Elsevier_Dec_2017}. In this context, the authors in \cite{Almasalma_Conference_June_2017} propose a dual-decomposition-based P2P algorithm to control voltage for the distribution network by actively managing the active and reactive power of DERs. An agent-based distributed power flow solver is studied in \cite{Nguyen_TSG_May_2015} to deal with problems from a completely distributed perspective. The study presented in \cite{Zizzo_TII_2018} focuses on the variation of power losses due to the superposition of P2P energy transactions in a microgrid. In particular, the authors propose using a blockchain for handling energy loss allocation and define a new timing for transacting intended P2P energy exchanges. Other examples of studies in this domain can be found in  \cite{Werth_TSG_July_2015} and \cite{Jogunola_Energies_2017}.

As evident from the above discussion, the surveyed studies have made significant contributions to the field of P2P trading. However, most of the discussed trading schemes have not emphasized the users' point of view on the adoption of such techniques. Note that consumer-centric design is important for the sustainable use of the techniques in the long run as pointed out by \cite{Tushar-TSG:2014}, \cite{Saad_ProcIEEE_2016} and \cite{Gangale_Policy_2013}. Therefore, in this work, we seek to complement the existing work by demonstrating how social cooperation between prosumers can lead to a consumer-centric energy trading mechanism by making the contributions discussed in Section~\ref{sec:introduction}. To this end, we begin by developing a suitable system model in the next section to use in the rest of the paper.

\section{System Description}\label{sec:system-model}
\begin{figure}
\centering
\includegraphics[width=\columnwidth]{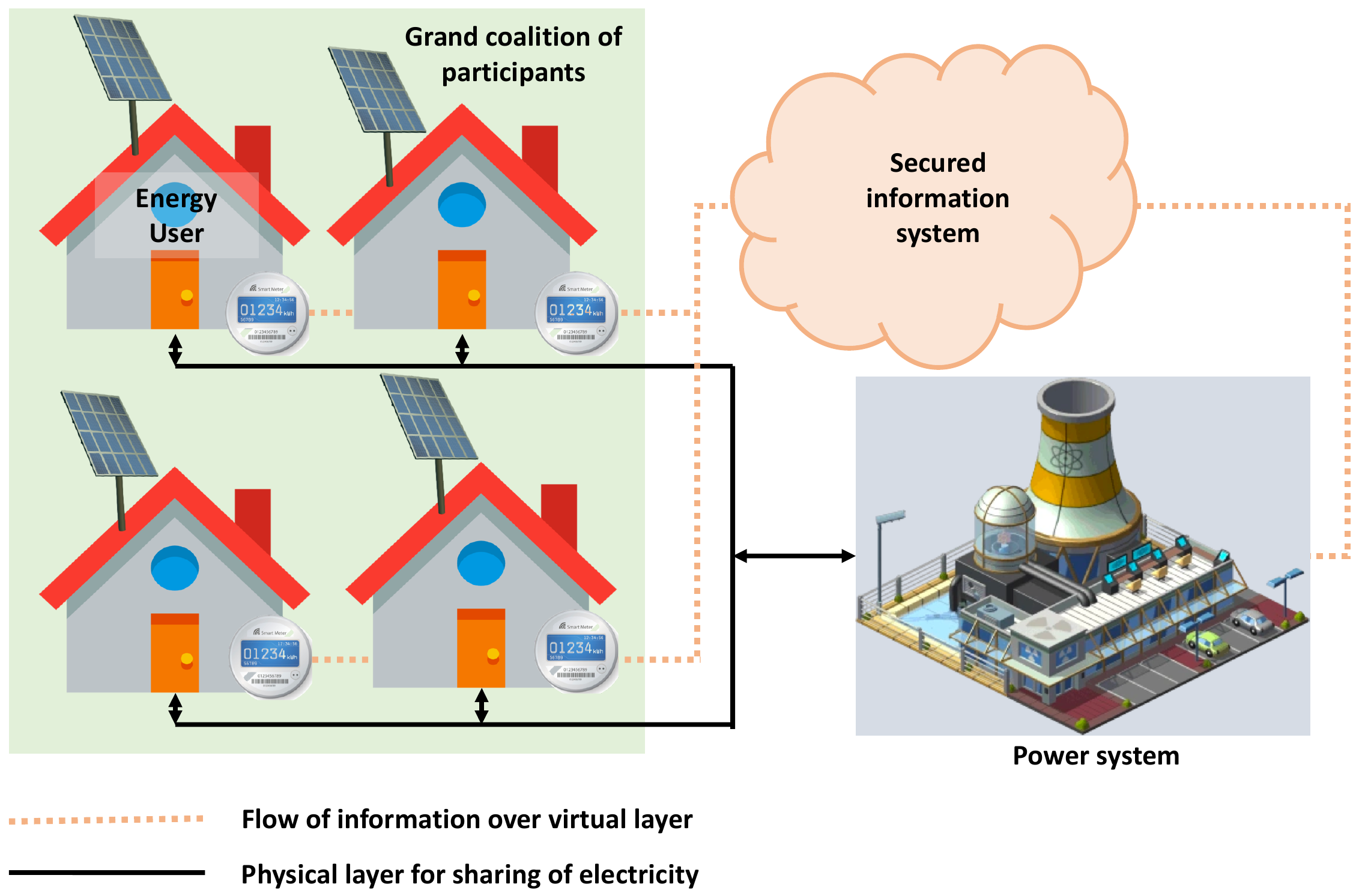}
\caption{This figure demonstrates the set up of an energy system in which the users can potentially form coalitions with one another to participate in the proposed P2P energy trading scheme.}
\label{fig:SystemModel}
\end{figure}
Let us assume a smart energy system consisting of a centralized power station (CPS) and $N$ prosumers, where $N = |\mathcal{N}|$. Each prosumer $n\in\mathcal{N}$ can be considered as a small user of energy in such a house, which is equipped with a rooftop solar panel without any storage facility\footnote{An example of such a system is the grid-tie solar system without a storage device~\cite{Tushar_TSG_July_2017}.} and can produce energy that it can either use for its household activities or trade with other energy entities within the energy network. Each prosumer $n\in\mathcal{N}$ is also equipped with a smart meter\footnote{Here, a smart meter refers to a smart hybrid system consisting of a number of necessary meters and devices as a single package to perform the relevant tasks of P2P trading.}, which is capable of determining how much energy the rooftop solar panel of the prosumer is generating, how much energy the prosumer is consuming, and how much energy the prosumer is selling to and/or buying from the CPS or any other entities, when necessary. 

The energy network under consideration is divided into two layers including a physical layer and a virtual layer.  In the \textit{physical layer}, all prosumers are connected to the energy network via a traditional distribution network (set up and run by the independent system operator), through which the physical exchange of energy takes place between the prosumers and the CPS, as well as the energy exchange between different prosumers within the distribution network. Examples of studies of the physical layer include \cite{Zeraati_TSG_EA_2016,Rijcke_TPS_Sept_2015,Zizzo_TII_EA_2018,Olek_TIE_Apr_2015}. However, the focus of this study is the virtual layer.

The exchange of all information is conducted over the \textit{virtual layer} of the network. Studies related to this layer focus on understanding the effect of economic factors on the customer's decision making process, for example, modeling customer behavior \cite{Navid_TSG_Nov_2015}, designing incentives and pricing schemes \cite{Tushar_TSG_July_2017}, modeling interactions between buyers and sellers \cite{Lee_TIE_June_2015}, and designing a consumer-centric scheme \cite{Rahi_TSG_EA_2018}. The work presented in this paper is exclusively relevant to this layer with the purpose of understanding the effects of economic factors, that is, pricing and the subsequent utility to the prosumers, on motivating extensive prosumer participation in energy trading. A common example of the virtual layer is the blockchain based layer~\cite{PowerLedger_2017}. Blockchain is being used by a number of P2P energy trading projects that are capable of distinguishing between energy transactions metered at the service entrance of each user and normal transactions with the grid (the utility, the supplier of electricity) or P2P transactions with other users, in order to bill correspondingly. An overview of the proposed system is shown in Fig.~\ref{fig:SystemModel}.

Let us consider that at any time during a day, each prosumer $n\in\mathcal{N}$, has an energy demand of $E_{n,d}$, and the energy generated from its rooftop solar panel is $E_{n,\text{pv}}$. Now, due to the fact that the energy generated from the solar panel is free of charge, it is reasonable to assume that each prosumer $n$ preferably consumes energy from its own rooftop solar panel to meet its own demand. Thus, the amount of energy consumed by the prosumer $n$ from its own generation can be expressed as
\begin{eqnarray}
E_{n,c} = \min(E_{n,d}, E_{n,\text{pv}}).
\label{eqn:ConsumedEnergy}
\end{eqnarray}
Now, depending on the values $E_{n,c}, E_{n,d}$ and $E_{n,\text{pv}}$, a prosumer can act either as a seller to sell its surplus energy $E_{n,\text{sur}}$ to or as a buyer to buy its deficit energy $E_{n,\text{def}}$ from the CPS or other energy entities within the network. Here,
\begin{eqnarray}
E_{n,\text{sur}} &=& E_{n,\text{pv}} - E_{n,c},\label{eqn:SurplusEnergy}~\text{and}\\
E_{n,\text{def}} &=& E_{n,d} - E_{n,c}.\label{eqn:DeficientEnergy}
\end{eqnarray}
We assume that all sellers belong to the set $\mathcal{N}_s$ and all buyers belong to the set $\mathcal{N}_b$. Clearly, $\mathcal{N}_s\cup\mathcal{N}_b = \mathcal{N}$ and $\mathcal{N}_s\cap\mathcal{N}_b = \phi$ at any particular time. Now, if the buying and selling prices per unit of energy are $p_b$ and $p_s$ respectively, the cost $J_n$ and revenue $U_n$ to each prosumer $n$ from such energy trading  is
\begin{eqnarray}
J_n = p_b E_{n,\text{def}},~\forall n\in\mathcal{N}_b,~\text{and}~U_n = p_s E_{n,\text{sur}},~\forall n\in\mathcal{N}_s,
\label{eqn:RevenueAndCost}
\end{eqnarray}
respectively. In the traditional energy market, such energy trading is usually conducted between the CPS and the prosumers. That is, each prosumer $n$ buys its energy $E_{n,\text{def}}$ from the CPS at a price $p_{b,g}$ per unit of energy, which is set by the CPS. Similarly, a prosumer sells its surplus energy $E_{n,\text{sur}}$, if there is any, to the CPS with a selling price $p_{s,g}$ defined by the CPS. Unfortunately, in general $p_{s,g}<<p_{b,g}$~\cite{Tushar-TIE:2015}, and consequently $J_n>>U_n$. Therefore, the monetary benefit that a prosumer obtains from its trading with the grid is very limited. As a consequence, there is a recent push towards a change in the operation style of the energy market, in which prosumers with energy surplus may trade the energy with another prosumer with energy deficiency via P2P energy trading within the same and/or neighboring energy network~\cite{Peck_Spectrum_2017}.

Nonetheless, to make such P2P energy trading a reality and to be a part of the energy market, it also needs to be consumer-centric, and prosumers need to find it beneficial to accept, so as to continuously participate in such trading. Otherwise, if P2P is not beneficial, they may withdraw prosumers with independent generation and storage capacity will have an incentive to go off-grid bringing an inefficient outcome, both for prosumers and the energy network~\cite{Morstyn_NatureEnergy_2018}. As such, we study a coalition game structure to show the potential of social cooperation among prosumers to develop a consumer-centric P2P energy trading framework in the next section.

\section{Coalition Game for P2P Trading Scheme}\label{sec:CoalitionGame}
A coalition game provides analytical tools to study the behavior of rational players when they cooperate~\cite{Saad_SPM_2009}. To design the proposed P2P energy trading under a coalition game framework, we consider the energy use of the entire energy network in the P2P energy trading paradigm. Therefore, the total amount of surplus energy that is available to all prosumers $\forall n\in\mathcal{N}_s$, after meeting their own demand is
\begin{eqnarray}
\sum_{n\in\mathcal{N}_s}E_{n,\text{sur}}=\sum_n^{N_s} E_{n,\text{pv}} - \sum_n^{N_s} E_{n,c},
\label{eqn:totalSurplus}
\end{eqnarray}
where $\sum_n^{N_s} E_{n,c} = \min\left(\sum_n^{N_s} E_{n,\text{pv}}, \sum_n^{N_s} E_{n,d}\right)$.
Similarly, the total energy deficiency of the prosumers $\forall n\in\mathcal{N}_b$ within the energy network is
\begin{eqnarray}
\sum_{n\in\mathcal{N}_b} E_{n,\text{def}}=\sum_n^{N_b} E_{n,d} - \sum_n^{N_b} E_{n,c}.
\label{eqn:totalDeficiency}
\end{eqnarray}
Indeed, the prosumers can sell their surplus energy to or buy their deficient energy from the CPS. Alternatively, the prosumers can also trade the energy among themselves via P2P energy trading considering the limited economic benefit from the energy trading with the CPS~\cite{Tushar-TIE:2015}. Nonetheless, the establishment of such P2P trading is contingent on the benefit that the prosumers may attain from the energy trading. For instance, if the monetary benefit from such trading is not attractive, and/or requires significant computational power, the motivation for the prosumers to adopt such trading scheme could be very low. Furthermore, for sustainable and consumer-centric P2P trading, it is also important that benefit to the prosumers is always better than the trading with the CPS every time they participate in the P2P trading~\cite{Hockenbury:2003}. 

In this context, we propose a CCG structure in the following section to demonstrate the benefits of forming a coalition among the prosumers for participating in P2P energy trading. Then, after discussing the properties of the game, we focus on how the proposed CCG based P2P energy trading structure establishes itself as a consumer-centric energy trading technique.
\subsection{Game Formulation} A CCG is characterized by a set of players $\mathcal{N} = \mathcal{N}_s\cup\mathcal{N}_b$ that forms a coalition, and a value function $\nu$ that demonstrates the worth of the coalition in terms of a numerical value. To this end, the proposed coalition game can be formally defined by its strategic form as
\begin{eqnarray}
\Gamma = \lbrace\mathcal{N},\nu\rbrace.\label{eqn:FormalDefinition}
\end{eqnarray}
Here, $\nu$ refers to the monetary amount that the participating prosumers, as a coalition group, may earn or spend during the P2P trading process. Hence, the proposed $\Gamma$ is a coalition game with transferrable utility, where the value function $\nu$ can be expressed as:
\begin{eqnarray}
\small
\nu(\mathcal{N}_s\cup\mathcal{N}_b) = p_{s,g} \max\left(0,\left(\sum_{n\in\mathcal{N}_s}E_{n,\text{sur}}-\sum_{m\in\mathcal{N}_b} E_{m,\text{def}}\right)\right)\nonumber\\- p_{b,g} \max\left(0,\left(\sum_{m\in\mathcal{N}_b} E_{m,\text{def}}- \sum_{n\in\mathcal{N}_s}E_{n,\text{sur}}\right)\right).
\label{eqn:ValueFunction}
\end{eqnarray}
Clearly, from \eqref{eqn:ValueFunction}, when the energy surplus of the prosumers in $\mathcal{N}_s$ is more than the total demand of prosumers in $\mathcal{N}_b$, the excess energy is sold to the grid (to reduce energy wastage), and vice versa. In other words, all the prosumers in $\mathcal{N}$ cooperate with one another to trade the surplus energy among themselves as a first priority, and then interact with the grid, if necessary, to sell or buy the total excess or deficient energy respectively. 

We note that there is no guarantee of forming a stable grand coalition (a single coalition of all prosumers within the network) in a CCG. The effectiveness of such a coalition is only confirmed if it is always beneficial for the prosumers to form a grand coalition, rather than acting noncooperatively or forming disjoint coalitions~\cite{Tushar_ISGTThailand_2015}. To this end, it is important that $\Gamma$ fulfils a number of requirements~\cite{Saad_SPM_2009} that are necessary for the effective and sustainable operation of the proposed P2P trading. These requirements are summarized as follows:
\begin{itemize}
\item\textit{Benefit of cooperation:} Cooperation, that is, the formation of the grand coalition, is never detrimental to any of the involved prosumers. In other words, in a CCG no group of prosumers can benefit by leaving the grand coalition, that is, by acting non-cooperatively. This is associated with the property of superadditivity of the value function of the game.
\begin{definition}
Consider the CCG $\Gamma = (\mathcal{N},\nu)$ in \eqref{eqn:FormalDefinition}, where $\nu$ is the value function of the game. Now, two disjoint subsets $\mathcal{N}_a\subseteq\mathcal{N}$ and $\mathcal{N}_b\subseteq\mathcal{N}$ will only cooperate together and form a grand coalition if $\nu$ satisfies the property of superadditivity, and therefore the following inequality holds:
\begin{equation}
\nu(\mathcal{N}_a\cup\mathcal{N}_b)\geq\nu(\mathcal{N}_a) + \nu(\mathcal{N}_b).
\label{eqn:Definition}
\end{equation}
\label{def:Superadditivity}
\end{definition}
\item\textit{Stability of coalition:} The benefit (or, revenue) of a coalition needs to be distributed among the prosumers in such a way that no individual or subgroup of prosumers have any incentive to abandon the grand coalition for further benefit. The set of feasible allocation of such revenues is defined as the core~\cite{Saad_SPM_2009}.
\begin{definition}
Let $\mathbf{e}$ be the payoff vector of the revenues that each prosumer of the CCG $\Gamma$ obtains, and the revenue of each prosumer $n\in\mathcal{N}$ is indicated as $e_n$ where $e_n\in\mathbf{e}$. Then the core of the $\Gamma$ is defined as\textit{~\cite{Saad_SPM_2009}}:
\begin{equation}
\mathcal{C} = \{\mathbf{e}:\sum_{n\in\mathcal{N}}e_n = \nu\left(\mathcal{N}\right) \text{\textit{and}} \sum_{n\in\mathcal{S}}e_n\geq\nu(\mathcal{S}),\forall\mathcal{S}\subseteq\mathcal{N}\}.
\label{eqn:Core}
\end{equation}
\label{def:Core}
\end{definition}
If $\mathcal{C}$ of the game is \textit{non-empty}, there exists a feasible allocation of revenues among the participating prosumers, in which no group of prosumers has any incentive to leave the coalition. Hence, a \textit{stable} coalition is established. Nonetheless, one way to understand whether $\Gamma$ has a non-empty core is through using the \textit{Bondareva-Shapley theorem}~\cite{Shapley_NRL_1967}, which can be defined as follows:
\begin{definition}
According to the Bondareva-Shapley theorem, the core $\mathcal{C}$ of a CCG $\Gamma$ is non-empty, if and only if for every function $f(\mathcal{S})$, where $\forall n\in\mathcal{N}:\sum_{\mathcal{S}\in\mathcal{P}_n}f(\mathcal{S})=1$, and $0\leq f(\mathcal{S})\leq 1$, the following inequality holds:
\begin{equation}
\sum_{\mathcal{S}\in\mathcal{P}\setminus\phi}f(\mathcal{S})\nu(\mathcal{S})\leq\nu(\mathcal{N}).
\label{eqn:ShapelyTheorem}
\end{equation}
Here, $\mathcal{P}$ is the power set of $\mathcal{N}$, and $\mathcal{P}_n\subseteq\mathcal{P}$ that has $n$ as one of the elements in all subsets.
\label{def:ShapelyTheorem}
\end{definition}
\end{itemize}
\section{Properties of the CCG Inspired P2P Trading Scheme}\label{sec:properties} In this section, we investigate two particular properties of the proposed scheme in order to seek answers to the following two questions: 1) is the cooperation formed between different prosumers for P2P trading stable? and 2) is the proposed P2P energy trading consumer-centric?
\subsection{Property of stability}\label{sec:stability}To do so, first we examine the superadditivity property of the value function $\nu$ of the the game.
\begin{theorem}
The value function $\nu$ of the proposed game $\Gamma$, as described in \eqref{eqn:ValueFunction}, is superadditive.
\label{theorem:1}
\end{theorem}
\begin{proof}
To prove this theorem, first we assume that $\sum_{n\in\mathcal{N}_s}E_{n,\text{sur}} - \sum_{m\in\mathcal{N}_b}E_{m,\text{sur}} = z$. Therefore, from \eqref{eqn:ValueFunction}, the value function can be expressed as
\begin{equation}
\nu = p_{s,g}\max(0,z) - p_{b,g}\max(0,-z).\label{eqn:ValueFunction3}
\end{equation}
We note that \eqref{eqn:ValueFunction3} is a concave function. 

Second, we assume that the set $\mathcal{N}_b$ of buyers and the set $\mathcal{N}_s$ of sellers can further be broken down into subsets $\mathcal{N}_{b,1}$ and $\mathcal{N}_{b,2}$, and $\mathcal{N}_{s,1}$ and $\mathcal{N}_{s,2}$ respectively, where $\mathcal{N}_{b,1}\cup\mathcal{N}_{b,2} = \mathcal{N}_b$, $\mathcal{N}_{b,1}\cap\mathcal{N}_{b,2} = \phi$, $\mathcal{N}_{s,1}\cup\mathcal{N}_{s,2} = \mathcal{N}_s$,and $\mathcal{N}_{s,1}\cap\mathcal{N}_{s,2} = \phi$. Then, due to the linearity of $\nu$ it is reasonable to write 
\begin{eqnarray}
\footnotesize
\frac{1}{2}\nu\Bigg[\sum_{n\in\mathcal{N}_s}E_{n,\text{sur}} - \sum_{m\in\mathcal{N}_b}E_{m,\text{def}}\Bigg]=\nonumber\\
\nu\Bigg[\sum_{n\in\mathcal{N}_s}\frac{E_{n,\text{sur}}}{2} - \sum_{m\in\mathcal{N}_b}\frac{E_{m,\text{def}}}{2}\Bigg]\nonumber\\=
\nu\Bigg[\Bigg[\sum_{n\in\mathcal{N}_{s,1}}\frac{E_{n,\text{sur}} }{2}-\sum_{m\in\mathcal{N}_{b,1}}\frac{E_{m,\text{def}}}{2}\Bigg]\nonumber\\+\Bigg[\sum_{n\in\mathcal{N}_{s,2}}\frac{E_{n,\text{sur}} }{2}-\sum_{m\in\mathcal{N}_{b,2}}\frac{E_{m,\text{def}}}{2}\Bigg]\Bigg].
\label{eqn:theorem1}
\end{eqnarray}
due to the linearity of the value function $\nu$. Now, according to  \cite{Lee_JSAC_2014}, due to the concavity of $\nu$, \eqref{eqn:theorem1} can be expressed based on Jensen's inequality as
\begin{eqnarray}
\footnotesize
\frac{1}{2}\nu\Bigg[\sum_{n\in\mathcal{N}_s}E_{n,\text{sur}} - \sum_{m\in\mathcal{N}_b}E_{m,\text{def}}\Bigg]\geq\nonumber\\\frac{1}{2}\nu\Bigg[\sum_{n\in\mathcal{N}_{s,1}}E_{n,\text{sur}}-\sum_{m\in\mathcal{N}_{b,1}}E_{m,\text{def}}\Bigg]\nonumber\\+\frac{1}{2}\nu\Bigg[\sum_{n\in\mathcal{N}_{s,2}}E_{n,\text{sur}}-\sum_{m\in\mathcal{N}_{b,2}}E_{m,\text{def}}\Bigg].
\label{eqn:theorem1_2}
\end{eqnarray}
From \eqref{eqn:theorem1_2}, clearly $\nu$ decreases as the number of disjoint coalitions increases, and therefore the value function $\nu$ of the proposed $\Gamma$ is superadditive.
\end{proof}
Therefore, forming a grand coalition is always beneficial for all participating prosumers in $\Gamma$. Now, we investigate if the grand coalition is stable, which is affected by the revenue that each prosumer attains by participating in the CCG $\Gamma$. For instance, if the trading price $p_{tr}\in\{p_{\text{s},tr}, p_{\text{b},tr}\}$ of the CCG is very close to $p_{b,s}$, the prosumers in $\mathcal{N}_s$ will be very satisfied. However, the buyer of set $\mathcal{N}_b$ will not have any motivation to stay in the coalition as the purchase price per unit of energy is too close to the grid price. Similarly, if $p_{tr}\approx p_{s,g}$, the sellers of the game will not be encouraged to stay within the coalition. As such, it is imperative that there exists a trading price $p_{s,g}\leq p_{tr}\leq p_{b,g}$, which would produce a set of revenues that would make the coalition stable. In this context, now we state and prove the following theorem.
\begin{theorem}
Under the grid's current pricing scheme, where $p_{b,g}>p_{s,g}$, the proposed CCG $\Gamma$ possesses a non-empty core when the trading price $p_{tr}$ of the P2P energy trading is within the range $p_{s,g}\leq p_{tr}\leq p_{b,g}$.
\label{theorem:2}
\end{theorem}
\begin{proof}
According to Definition~\ref{def:ShapelyTheorem}, to show that the core $\mathcal{C}$ of the game $\Gamma$ is non-empty, it is sufficient to show that $\nu$ satisfies \eqref{eqn:ShapelyTheorem}. This can be proven by following the same procedure explained in Theorem~2 in \cite{Lee_JSAC_2014} (on page 1401).
\end{proof}
\subsubsection{Distribution of revenue}Now, according to Theorem~\ref{theorem:2}, it is clear that the value function $\nu$ of the proposed $\Gamma$ has a non-empty core, in which each participating prosumer $n\in\mathcal{N}$ receives a revenue that makes the coalition stable. To attain the revenue distribution, different techniques including shapley value~\cite{Lee_JSAC_2014}, nucleolus~\cite{Langer_TMC_2015}, and proportional fairness~\cite{Zhang_TWC_2014} have been used in the literature. In this paper, however, we propose to use the mid-market rate pricing~\cite{Long_Conf_2017} as a means of distributing the revenue between the prosumers. The main reasons behind using this technique are: 1) this is simple to implement without any computational complexity, which is important and expected for the practical deployment of P2P trading scheme~\cite{ChrisGiotitsas_TS_2015}, 2) its suitability for energy trading has been demonstrated by its deployment in a P2P energy trading testbed in Europe~\cite{Long_Conf_2017}, and 3) importantly, the pricing scheme ensures that the core of the proposed value function is non-empty, and therefore we achieve a stable coalition by using this mid-market rate.
\subsubsection*{Mid-market rate}\label{sec:MidMarketRate} According to mid-market rate, the price per unit of energy for energy trading is decided based on three different cases~\cite{Long_Conf_2017}: 1) Generation is equal to demand, 2) Generation is greater than demand, and 3) Generation is lower than demand. Nevertheless, in all cases, the energy trading \textit{between} the participants takes place with the price $p_{tr}$ per unit of energy, which is chosen to be the mid-value of the buying and selling prices set by the grid for its trading with the EUs. 
\begin{equation}
p_{tr} = \frac{p_{s,g} + p_{b,g}}{2}.
\label{eqn:TradingPriceP2P}
\end{equation}

\textit{\textbf{Case 1 - Generation is equal to demand:}} In this case, the net energy demand and production of all prosumers within the network is zero. That is, the total surplus energy $\sum_{n\in\mathcal{N}_s}E_{n,\text{sur}}$ of $N_s$ sellers is sold to the buyers of set $\mathcal{N}_b$. Hence, the selling price $p_{\text{s},tr}$ and buying price $p_{\text{b},tr}$ of each participant in $n\in\mathcal{N}_s$ and $m\in\mathcal{N}_b$ respectively are equal to one another and as same as the expression in \eqref{eqn:TradingPriceP2P}.

\textit{\textbf{Case 2 - Generation is greater than demand:}} In this case, the net energy production is non-zero, and therefore the sellers can sell the total surplus energy to the grid at a price $p_{s,g}$ per unit of energy after meeting the demand of prosumers with energy deficiency of the network. Clearly, the buying price $p_{\text{b},tr}$ of each buyer is
\begin{equation}
p_{\text{b},tr} = p_{tr} = \frac{p_{s,g} + p_{b,g}}{2}.
\label{eqn:GenerationMore1}
\end{equation}
The selling price $p_{\text{s},tr}$ per unit of energy in this case, however, depends on the total generation $\sum_{n\in\mathcal{N}_s}E_{n,\text{pv}}$, total demand $\sum_{n\in\mathcal{N}_b}E_{n,d}$ of prosumers with energy deficiency, and prices $p_{tr}$ and $p_{s,g}$. In particular, $p_{\text{s},tr}$ can be expressed as
\begin{equation}
\footnotesize
p_{\text{s},tr}=\frac{\displaystyle\sum_{m\in\mathcal{N}_b}E_{m,d}\times p_{tr} + \left(\sum_{n\in\mathcal{N}_s}E_{n,\text{sur}}-\sum_{m\in\mathcal{N}_b}E_{m,d}\right)\times p_{s,g}}{\sum_{n\in\mathcal{N}_s}E_{n,\text{sur}}}.
\label{eqn:GenerationMore2}
\end{equation}
In \eqref{eqn:GenerationMore2}, the numerator refers to the total revenue that the sellers of the network can earn by selling their surplus energy. The first term $\displaystyle\sum_{m\in\mathcal{N}_b}E_{m,d}\times p_{tr}$ denotes the revenue by selling energy to the buyers in $\mathcal{N}_b$ at a price $p_{tr}$ per unit of energy, and $\left(\sum_{n\in\mathcal{N}_s}E_{n,\text{sur}}-\sum_{m\in\mathcal{N}_b}E_{m,d}\right)\times p_{s,g}$ is the revenue gained from selling the rest of energy to the grid.

\textit{\textbf{Case 3 - Generation is lower than demand:}} In this case, the net energy demand within the network is non-zero. Therefore, the buyers of $\mathcal{N}_b$ need to meet their excess energy demand $\sum_{m\in\mathcal{N}_b}E_{m,d} - \sum_{n\in\mathcal{N}_s}E_{n,\text{sur}}$ from the grid. Indeed, as in case 1 and 2, the selling price $p_{\text{s},tr}$ that each seller $n\in\mathcal{N}_s$ charges the buyers for selling their surplus energy is equal to $p_{\text{s},tr}=p_{tr} = \frac{p_{s,g} + p_{b,g}}{2}$. The buying price $p_{\text{s},tr}$ per unit of energy, on the other hand, will be affected by the available total surplus $\sum_{n\in\mathcal{N}_s}E_{n,\text{sur}}$, total demand $\sum_{m\in\mathcal{N}_b}E_{m,d}$ and the prices $p_{tr}$ and $p_{b,g}$. That is
\begin{equation}
\small
p_{\text{b},tr} = \frac{\displaystyle\sum_{n\in\mathcal{N}_s}E_{n,\text{sur}}\times p_{tr}+\left(\sum_{m\in\mathcal{N}_b}E_{m,d} - \sum_{n\in\mathcal{N}_s}E_{n,\text{sur}}\right)\times p_{b,g}}{\sum_{m\in\mathcal{N}_b}E_{m,d}}.
\label{eqn:DemandMore1}
\end{equation}
Here, $\sum_{n\in\mathcal{N}_s}E_{n,\text{sur}}\times p_{tr}$ is the cost to the buyers for buying energy from the prosumers with energy surplus, and $\left(\sum_{m\in\mathcal{N}_b}E_{m,d} - \sum_{n\in\mathcal{N}_s}E_{n,\text{sur}}\right)\times p_{b,g}$ is the cost of buying the rest of the need from the grid.

Here, it is important to the note that, as the P2P trading is proposed, depending on the values of $\sum_{m\in\mathcal{N}_b}E_{m,d}$, $\sum_{n\in\mathcal{N}_s}E_{n,\text{sur}}$, $p_{b,g}$ and $p_{s,g}$, the values of  $p_{s,tr}$ and $p_{b,tr}$ are fixed for each time slot irrespective of whether the prosumers form a grand coalition or disjoint coalition. In other words, once a prosumer decides to trade energy with other prosumers instead of trading with the grid, it needs to buy and sell using $p_{b,tr}$ and $p_{s,tr}$ respectively set for that time slot. Note that we do not consider the regulatory charges within the pricing scheme. However, the P2P platform provider can charge the prosumers a fee incorporated in the trading price \cite{Yu_IoTJ_Dec_2014} and then pay the ISO a subscription fee for using its infrastructure for P2P trading~\cite{Tushar_SPM_July_2018}.

\begin{theorem}
For the considered mid-market rate pricing schemes in Case-1, Case-2 and Case-3, the core of the proposed CCG $\Gamma$ is non-empty, and therefore the formation of a stable grand coalition is confirmed.
\label{theorem:3}
\end{theorem}
\begin{proof}
To prove this theorem, first we note from Theorem~\ref{theorem:2} that a grand coalition for the proposed CCG $\Gamma$ is always stable for a pricing scheme, in which the trading price, which includes both selling and buying prices $p_{\text{s},tr}$ and $p_{\text{b},tr}$ respectively, satisfies the condition $p_{s,g}\leq\{p_{\text{s},tr},p_{\text{b},tr}\}\leq p_{b,g}$. Second, it is clear from the description of above three cases that both trading prices in Case 1, the buying price $p_{\text{b},tr}$ in Case 2, and the selling price $p_{\text{s},tr}$ in Case 3 satisfy the above mentioned condition. Therefore, proving that both \eqref{eqn:GenerationMore2} and \eqref{eqn:DemandMore1} also satisfy the conditions $p_{s,g}\leq p_{\text{s},tr}\leq p_{b,g}$ and $p_{s,g}\leq p_{\text{b},tr}\leq p_{b,g}$ respectively is sufficient to complete the proof of Theorem~\ref{theorem:3}.

As such, let us first assume that $\frac{\sum_{m\in\mathcal{N}_b}E_{m,d}}{\sum_{n\in\mathcal{N}_s}E_{n,\text{sur}}} = k$, where $k<1$ for Case 2 (as $\sum_{n\in\mathcal{N}_s}E_{n,\text{sur}}>\sum_{m\in\mathcal{N}_b}E_{m,d}$), and based on this assumption \eqref{eqn:GenerationMore2} can be written as 
\begin{equation}
p_{\text{s},tr} = k\times p_{tr}+(1-k)\times p_{s,g} = (k\times p_{tr} + p_{s,g}) - k\times p_{s,g}.
\label{eqn:theorem3_1}
\end{equation}
Now, from \eqref{eqn:TradingPriceP2P}, clearly $p_{tr}>p_{s,g}$ as $p_{b,g}>p_{s,g}$. Hence, from \eqref{eqn:theorem3_1}, we can confirm that $p_{\text{s},tr}\geq p_{s,g}$. Now, to prove that $p_{\text{s},tr}\leq p_{b,g}$, first we consider that $p_{\text{s},tr}>p_{b,g}$, and therefore, from \eqref{eqn:theorem3_1}
\begin{equation}
k\times p_{tr} + p_{s,g} - k\times p_{s,g}>p_{b,g}.
\label{eqn:theorem3_2}
\end{equation}
Then, replacing $p_{tr}$ with $\frac{p_{s,g} + p_{b,g}}{2}$, and rearranging the terms, \eqref{eqn:theorem3_2} can be expressed as
\begin{equation}
p_{s,g} - \frac{k}{2}p_{s,g}>p_{b,g}-\frac{k}{2}p_{b,g},
\label{eqn:theorem3_3}
\end{equation}
which is not possible as $p_{b,g}>p_{s,g}$ and $k<1$. Hence, $p_{\text{s},tr}\leq p_{b,g}$. So, $p_{\text{s},tr}$ in \eqref{eqn:GenerationMore2} satisfies the condition $p_{s,g}\leq p_{\text{s},tr}\leq p_{b,g}$.

Similarly, by assuming $\frac{\sum_{n\in\mathcal{N}_s}E_{n,\text{sur}}}{\sum_{m\in\mathcal{N}_b}E_{m,d}} = k'$ in \eqref{eqn:DemandMore1}, and following the same procedures as described for $p_{s,g}$ in \eqref{eqn:GenerationMore2}, it can be proven that $p_{\text{b},tr}$ in \eqref{eqn:DemandMore1} also satisfies the condition $p_{s,g}\leq p_{\text{b},tr}\leq p_{b,g}$, and thus Theorem~\ref{theorem:3} is proven.
\end{proof}
\begin{remark}
An underlying assumption in the proposed scheme is that only the grid and the designed trading platform that provides the P2P trading services are in the considered system. Hence, there is no other competitor that offers different services for P2P trading. However, if more competitors exist, the game needs to be designed in a different manner, which is an interesting extension for future work. Nonetheless, if there is no competitor, the grand coalition proposed in this study is a stable coalition.
\end{remark}
\subsection{Consumer-Centric Property}\label{sec:PeopleCentric}As we mentioned in Section~\ref{sec:introduction}, for the successful establishment of an energy trading scheme, it is critical that the prosumers within the system are actively participating in the trading. In this context, the purpose of this section is to understand how people can be motivated to participate in energy trading, and then investigate whether the proposed energy trading scheme in this paper fulfills the requirements of the motivational models, and thus can be considered as consumer-centric.

To this end, first, we note that motivation is closely related to the emotional process that initiates behavior~\cite{Hockenbury:2003}. It is a branch of behavioural science that helps us to understand how to mediate the psychological process that guides real behavior~\cite{Beebe:1999,Miller:2002}, and has been used in engineering~\cite{Franca_IST:2014,Beecham_IST:2008}, public health~\cite{Glanz_ARPH:2010}, education system~\cite{Hermida_TNE:2015}, economics~\cite{Stein_IO:2017}, and medicine~\cite{Carr_IO:2017}. In motivational psychology, there are several models that can be studied to show how users can be motivated to adopt a certain behavior. Examples of such behavioural models, which are also the subject of this section, include the rational-economic model and the positive reinforcement model.
\begin{definition}
The rational-economic model establishes that people adopt pro-environmental behavior based on economically rational decisions~\cite{Shipworth:2000}. In other words, monetary cost is the key motivator for people to take necessary actions, e.g., participating in the P2P energy trading. 
\label{def:RationalEconomic}
\end{definition}
\begin{definition}
 A positive reinforcement defines the case when a human response to a situation is followed by a reinforcing stimulus that increases the possibility of having the same response from them when a similar situation occurs~\cite{Hockenbury_2003}. 
 \label{def:PositiveReinforcement}
\end{definition}
According to these definitions, a trading scheme that satisfies both rational-economic and positive reinforcement properties has very high possibility to be accepted by the customers in the market, and hence would be a consumer-centric scheme. Therefore, to confirm that the proposed trading scheme is consumer-centric, it is sufficient to show that the proposed scheme satisfies both the models. 

Now, clearly, for the considered mid-market rate pricing scheme, the proposed $\Gamma$ based P2P energy trading scheme satisfies the rational-economic model due to the following facts 
\begin{itemize}
\item The value of coalition in \eqref{eqn:ValueFunction} is defined in terms of \textit{monetary revenue} that the coalition attains from participating in the proposed P2P energy trading.
\item In Theorem~\ref{theorem:2}, it is shown that the core of the proposed $\Gamma$ is non-empty. Therefore, there exists a revenue distribution vector for the participating prosumers such that none of the prosumers would have any incentive to leave the grand coalition.
\item Finally, in Theorem~\ref{theorem:3}, it is proven that the revenue that each prosumer obtains by using the mid-market rate pricing scheme in the proposed $\Gamma$  lies within the core of the game.
\end{itemize}
Therefore, all the participating prosumers in the proposed $\Gamma$ are always satisfied with the monetary revenues that they receive by participating in the proposed P2P.  
\begin{theorem}
The discussed P2P energy trading technique based on the proposed $\Gamma$ complies with the property of positive reinforcement.
\label{theorem:5}
\end{theorem}
\begin{proof}
According to Definition~\ref{def:PositiveReinforcement}, the positive reinforcement property confirms that a prosumer will get positive encouragement to do anything, for example, participating in the P2P energy trading in this case, if he receives the similar positive outcome every time he participates. Now, we note that at any time of a day, depending on the demand and supply of energy to the prosumers,  the proposed CCG is conducted for any of the three cases mentioned in Section~\ref{sec:MidMarketRate}. Now, it is proven in Theorem~\ref{theorem:3} that for each of the three cases the revenue distribution always lies within the core of the game. In other words, regardless of the type of the case, the prosumers benefit every time they participate in the proposed P2P energy trading. Thus, the proposed P2P energy trading satisfies the positive reinforcement model and the theorem is proven.
\end{proof}
Therefore, it is reasonable to say that the proposed P2P energy trading scheme is a consumer-centric scheme.

\section{Case Study}\label{sec:CaseStudy}
To show the effectiveness of social cooperation in the proposed P2P energy trading scheme, in this section we demonstrate some results from numerical case studies. In particular, we show that the proposed P2P energy trading scheme 1) ensures the formation of a stable coalition among the participating prosumers, 2) brings benefits to the participants in terms of reducing overall energy usage cost compared to the non-participating prosumers, and 3) satisfies the consumer-centric property. For the numerical study, we use publicly available real-data on solar generation\footnote{Public solar data is collected from IEEE PES ISS website.} and household energy demand of residential consumers\footnote{Residential data is available in the website of National Energy Efficiency Alliance (NEEA).}.  We consider five residential houses as prosumers, each of which is equipped with a $3$kWp solar panel. We use $15$ minute resolution data to validate the model, and the data used for this case study was recorded in December $2013$. The values of grid's buying price ($p_{s,g}$) and FiT price ($p_{b,g}$) are assumed to be $24.6$ and $10$ cents/kWh respectively according to the electricity price in Brisbane, Australia.
\begin{figure}[t]
\centering
\includegraphics[width=\columnwidth]{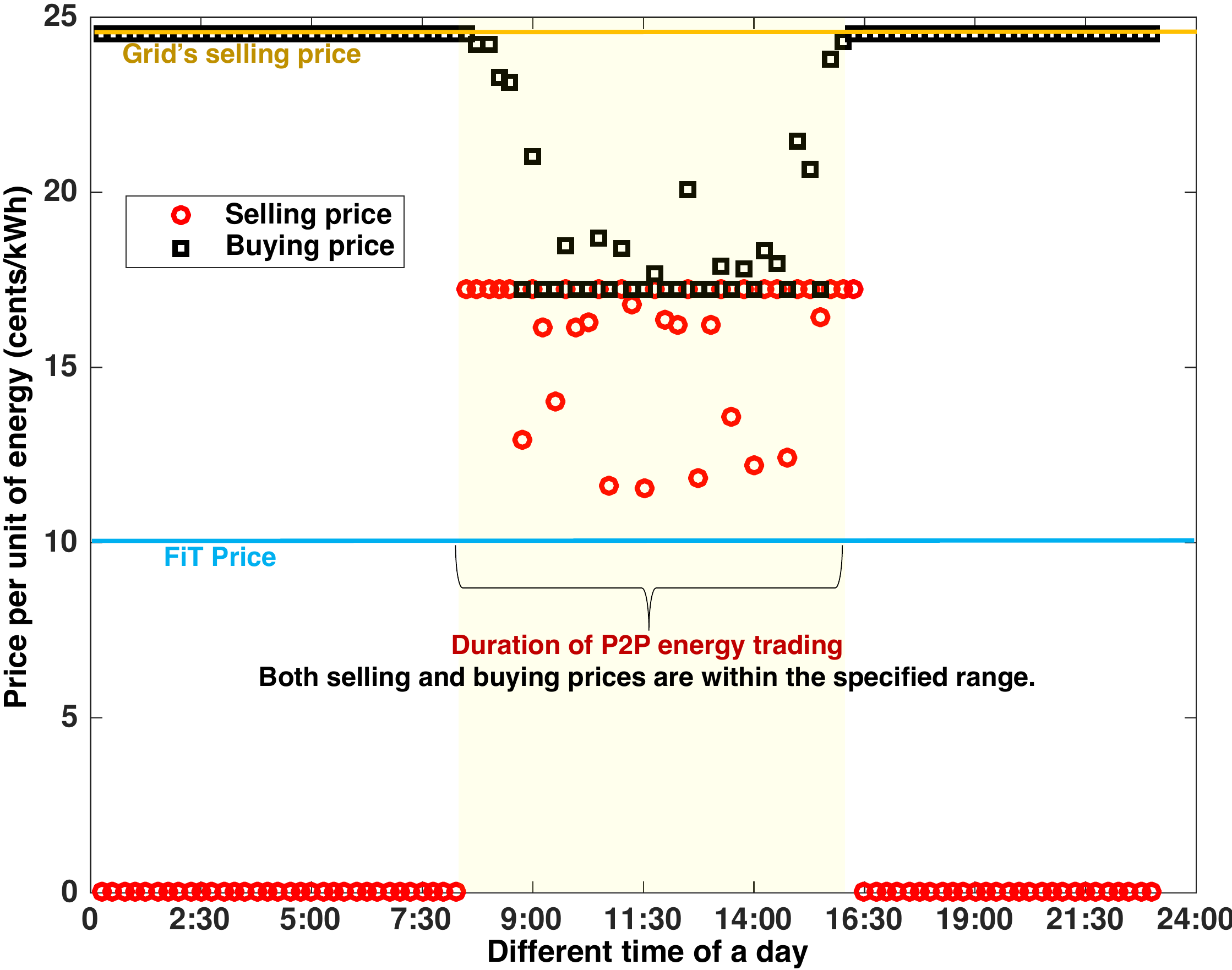}
\caption{This figure demonstrates the buying and selling price per unit of energy at different time slots of the day. The price during P2P energy trading is only visible when solar energy is available.}
\label{fig:result:1}
\end{figure}
\subsubsection{Formation of stable coalition}We note that the coalition between the participating prosumers in the proposed CCG is stable if the CCG possesses a non-empty core. From Theorem~\ref{theorem:2}, the condition of having a non-empty core is that the trading prices including both the buying and selling prices per unit of energy (i.e., $p_{\text{b},tr}$ and $p_{\text{s},tr}$ respectively) during the P2P energy trading always need to be within the range $\left[p_{s,g}, p_{b,g}\right]$. In this context, we show the trading price per unit of energy for a single day (December $2$, 2013) in Fig.~\ref{fig:result:1}. From the figure, first we note that P2P trading only takes place from $8.00$ am to $3.00$ pm as that is the duration of time when the sun was available to produce energy from the prosumers' solar panels. Consequently, for the rest of the day, prosumers need to rely on the power from the CPS and do not cooperate with one another. Second, during the P2P trading period, it is obvious from Fig.~\ref{fig:result:1} that the trading prices, which are developed based on the mid-market rate in \eqref{eqn:TradingPriceP2P}, \eqref{eqn:GenerationMore1}, \eqref{eqn:GenerationMore2}, and \eqref{eqn:DemandMore1}, are always within the specified range between the grid's selling price and FiT price. Hence, every time the households participate in the proposed P2P energy trading scheme, they form a stable grand coalition for trading energy with one another to maximize their benefits in terms of cost saving.
\begin{table}[t]
\centering
\caption{This table demonstrates the total cost to different prosumers .}
\includegraphics[width=\columnwidth]{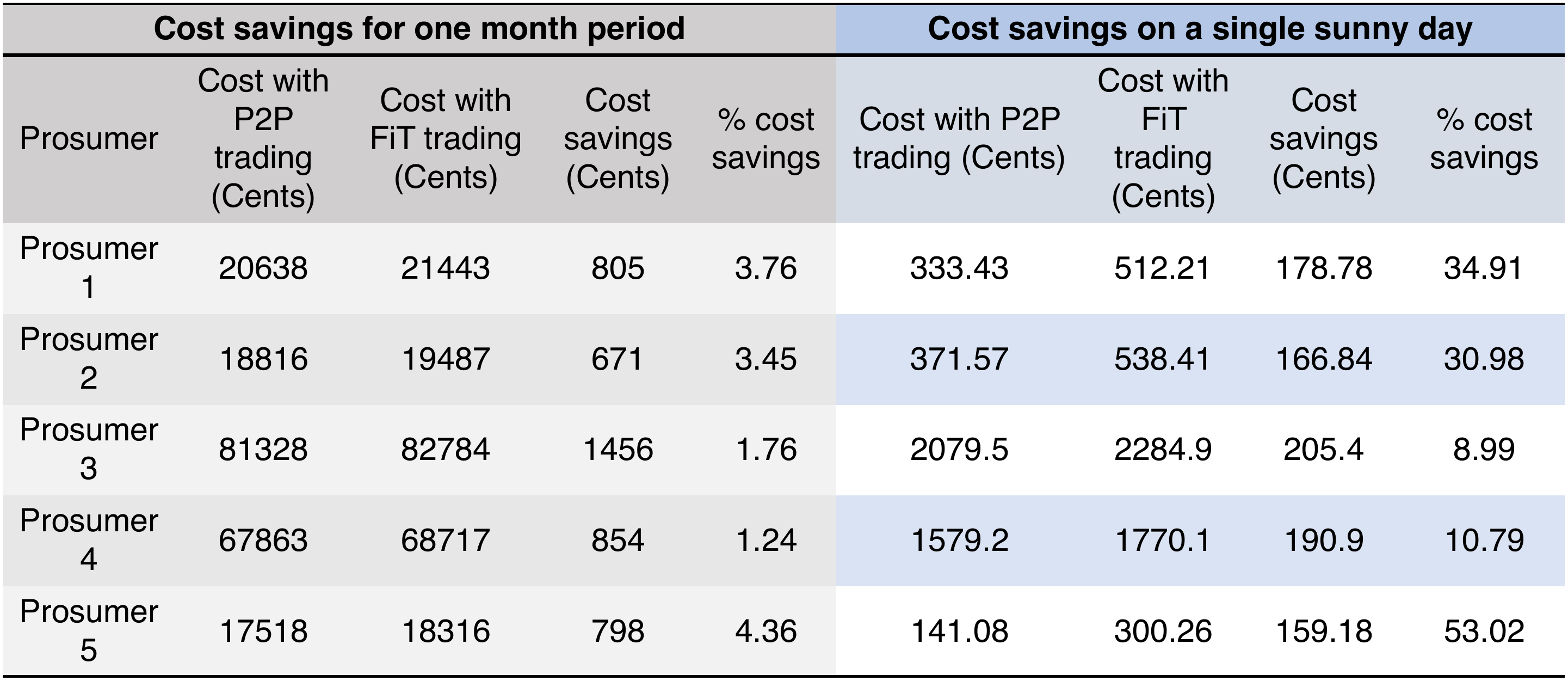}
\label{table:table1}
\end{table}
\subsubsection{Cost saving to each prosumer} To demonstrate how the proposed P2P trading may help each prosumer to reduce its cost of energy usage, we show the total cost to each of the five prosumers for a month (December 2013) in Table~\ref{table:table1}. The demonstrated costs include both the cost for adopting the proposed P2P scheme as well as for using the traditional FiT scheme. Now, based on the information illustrated in Table~\ref{table:table1}, the P2P scheme always outperforms the current FiT scheme in terms of reducing the total cost to each prosumer. For example, prosumer $1, 2, 3, 4$ and $5$ can save around $\$8.05$, $\$6.71$, $\$14.56$, $\$8.54$, and $\$7.98$ or percentage savings of $3.76\%$, $3.45\%$, $1.76\%$, $1.24\%$, and $4.36\%$ respectively. Based on this result, we can summarize that percentage savings are different for different users. A lower percentage can be translated into relatively large monetary savings compared to other prosumers with higher percentage savings. For instance, although prosumer $3$ and prosumer $5$ have percentage savings of $1.76\%$ and $4.36\%$ respectively, their actual total monetary savings in the respective month are $\$14.56$ and $\$8.50$ respectively.

Further, we note that the percentage savings for the P2P scheme are not substantial compared to the FiT scheme as shown in Table~\ref{table:table1}, which is, in fact, a result of the lower sunshine hours at different days of the month. For example, we observe from the dataset that in multiple days of the month, there was zero production of energy from the solar panels, which subsequently increases the total cost to the prosumers significantly (thus reducing the savings) across the whole month. Nonetheless, the savings can significantly improve on sunny days. For example, on December 2, 2013 (in the right-hand side of the same table), the percentage cost saving to each prosumer varied from as low as $9\%$ to as high as $53\%$, which is a substantial saving. The difference in cost savings is, however, for the same reason mentioned previously. Nevertheless, it is obvious from this result that the social cooperation between the prosumers in the proposed CCG based P2P trading scheme has the potential to bring benefit to the prosumers compared to the case without any cooperation.

\begin{figure}[t]
\centering
\includegraphics[width=\columnwidth]{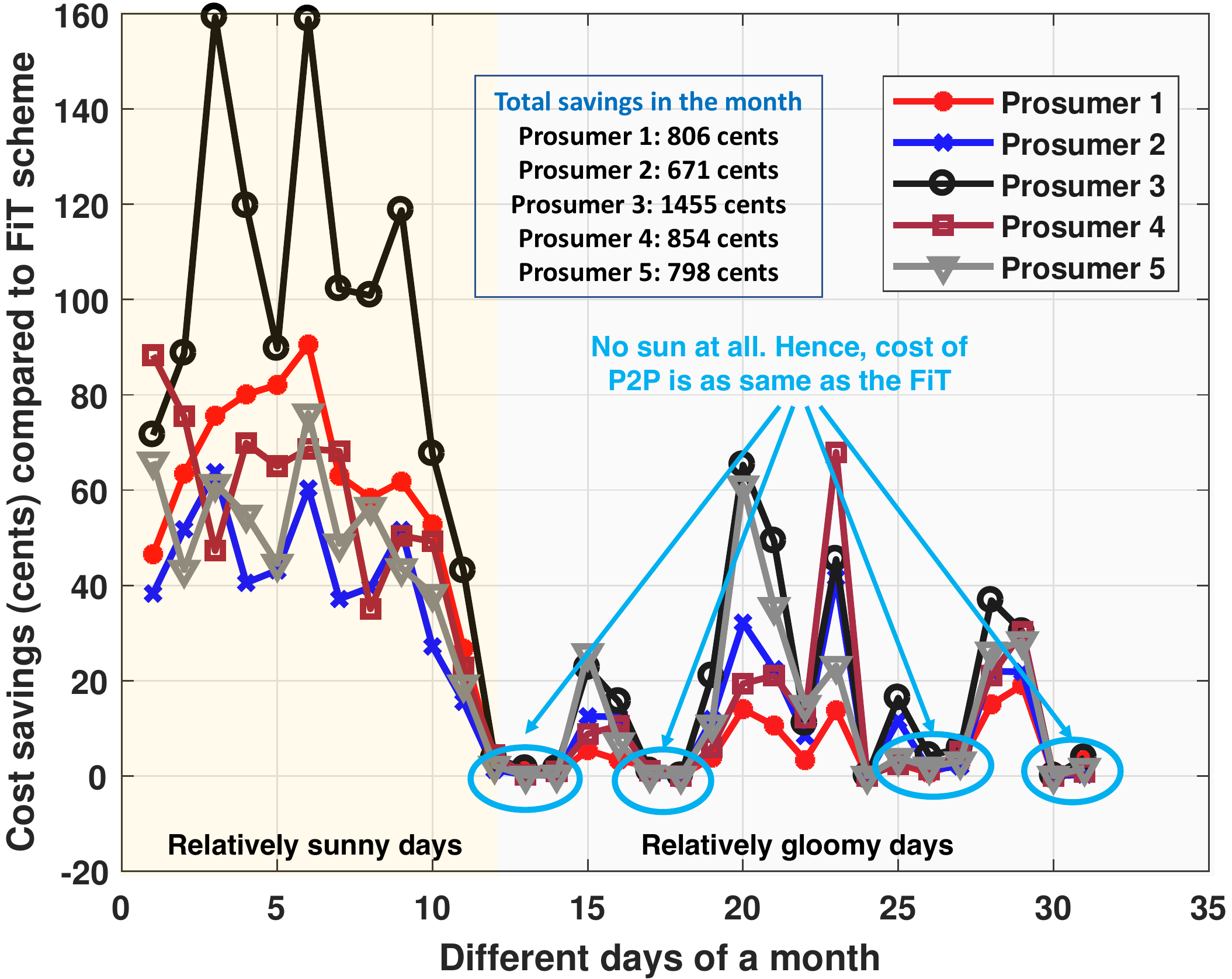}
\caption{This figure demonstrates how the proposed CCG satisfies the positive reinforcement property for the proposed pricing scheme.}
\label{fig:PositiveReinforcement}
\end{figure}
\subsubsection{Attainment of consumer-centric property}Finally, we show if the proposed P2P trading scheme demonstrates the consumer-centric property. To this end, first, we note that the rational-economic property of the proposed scheme is already demonstrated in the previous section, in which we show that the proposed scheme always brings economic benefit to each prosumer. Hence, to confirm that the proposed scheme possesses the consumer-centric property, it is sufficient to show that the proposed scheme also satisfies the positive-reinforcement property.

In this context, in Fig.~\ref{fig:PositiveReinforcement}, we show the benefit to five prosumers in terms of cost savings compared to the FiT scheme for each day of a month. As can be seen from the figure, social cooperation always benefits the prosumers as long as the day is not without sun. Of course, sunny days (as demonstrated by the earlier days of the month) benefit the prosumers more compared to days with relatively less sunshine time (demonstrated by the later days of the month). However, in the worst case when there is no sun in the sky, the cost savings compared to FiT is zero, that is, the cost of social cooperation is the same as the cost for participating in the FiT scheme. Nonetheless, such events are not very frequent as can be seen from the figure. Therefore, it can be concluded that the benefit to the houses for participating in P2P via social cooperation is \textit{consistent across each day the month} depending on the sunshine time each day and never detrimental (never increases the cost to the prosumers compared to the existing FiT). As a result, the overall cost savings across a month is also noticeable. For instance, as shown in Fig.~\ref{fig:PositiveReinforcement}, the total cost savings per month to prosumer 1, 2, 3, 4, and 5 are $\$8.00, \$6.70, \$14.50, \$8.50$ and $\$8.00$ respectively. Hence, it is reasonable to establish that the proposed P2P scheme always benefits the prosumers, and thus satisfies the positive reinforcement model. Hence, the proposed scheme is consumer-centric.
\section{Conclusion}\label{sec:Conclusion}In this paper, we have studied the feasibility of social cooperation among prosumers in confirming sustainable users' participation in P2P energy trading. To do so, a P2P energy trading scheme inspired by a CCG has been proposed and its different properties have been studied. It has been shown that the game possesses a non-empty core, which confirms the stability of the grand coalition among the participating prosumers. Further, a mid-market rate based pricing scheme has been proposed and it has been demonstrated that the revenue that each prosumer received for this mid-market rate lies within the core. Furthermore, to confirm the sustainable participation of the prosumers in the proposed P2P energy trading scheme, we have introduced two models from motivational psychology, and have shown that the proposed P2P scheme satisfies both of them. Note that this consequently has proven the potential of users' acceptance of the proposed scheme. Finally, we have provided numerical case studies to prove the claims that we made in the paper. 

A potential extension of this work would be to incorporate network constraints such as voltage constraints, thermal constraints, and ramp rates into the designed model and identify how this impacts user participation in the P2P trading. Another future extension of the proposed work would be to investigate how the behavior of the system is affected if an integrated storage device with each of the prosumers' solar systems is considered. Further, another interesting future extension of the proposed work is to determine how the stability of the proposed coalition is affected when there are multiple P2P energy trading platform providers in the network offering different pricing schemes for trading.

\end{document}